\documentclass[onecolumn]{autart}    % Enable this line and disable the 
\newtheorem{proof}{Proof}[section]

\usepackage{amsmath, amssymb,mathrsfs,amsfonts,dsfont, graphicx,url,gensymb,algorithm,nccmath}
\let\classAND\AND
\let\AND\relax
\usepackage{algorithmic}
\let\AND\classAND
\AtBeginEnvironment{algorithmic}{\let\AND\algoAND}
\usepackage{multirow, caption, subcaption,bm}
\usepackage[normalem]{ulem}
\usepackage[bb=boondox]{mathalpha}

 \usepackage[compact]{titlesec}
 \titlespacing{\section}{0pt}{*0}{*0}
 \titlespacing{\subsection}{0pt}{*0}{*0}
 \titlespacing{\subsubsection}{0pt}{*0}{*0}
 \AtBeginDocument{%
  \addtolength\abovedisplayskip{-0.2\baselineskip}%
  \addtolength\belowdisplayskip{-0.2\baselineskip}}%
\usepackage[dvipsnames]{xcolor}
\begin{document}

\begin{frontmatter}

\title{Approximate constrained stochastic optimal control via parameterized input inference}
% {Parameterized input inference for approximate structured stochastic optimal control} % Title, preferably not more 

\thanks[footnoteinfo]{The work was supported by the U.S. National Science Foundation (NSF) under Grant No. 1925147, 2212582, and 2241585.}
\thanks[footnoteinfo]{A part of this manuscript was published in Automatica Volume 171, January 2025.}
\thanks[footnoteinfo]{~~~Corresponding author: Shahbaz P Qadri Syed (shahbaz\_qadri.syed@okstate.edu)}
\author[okstate]{Shahbaz P Qadri Syed}\ead{shahbaz\_qadri.syed@okstate.edu},    % Add the 
\author[okstate]{He Bai}\ead{ he.bai@okstate.edu}

\address[okstate]{Mechanical and Aerospace Engineering, Oklahoma State University, Stillwater, OK, 74078}

                             % keyword list or with the 
                                          % help of the Automatica 
                                          % keyword wizard

\begin{abstract}                          % Abstract of not more than 200 words.
Approximate methods to solve stochastic optimal control (SOC) problems have received significant interest from researchers in the past decade. Probabilistic inference approaches to SOC have been developed to solve nonlinear quadratic Gaussian problems. In this work, we propose an Expectation-Maximization (EM) based inference procedure to generate state-feedback controls for constrained SOC problems. We consider the inequality constraints for the state and controls and also the structural constraints for the controls. We employ barrier functions to address state {and control} constraints. We show that the expectation step leads to smoothing of the state{-control} pair while the the maximization step on the non-zero subsets of the control parameters allows inference of structured stochastic optimal controllers. We demonstrate the effectiveness of the algorithm on unicycle obstacle avoidance, four-unicycle formation control, and quadcopter navigation in windy environment examples. In these examples, we perform an empirical study on the parametric effect of barrier functions on the state constraint satisfaction. We also present a comparative study of  smoothing algorithms on the performance of the proposed approach.
\end{abstract}
\end{frontmatter}

\section{Introduction}
Stochastic optimal control (SOC) is defined as the problem of finding a controller that minimizes an expected cost in the presence of uncertainty and dynamics constraint. The uncertainty is either in the form of noisy observations or process noise that approximates  model uncertainties in the system. A solution to the SOC problem can be found by solving the nonlinear stochastic Hamilton-Jacobi-Bellman (HJB) equation~\cite{oce}. In general, its numerical solution is computationally intractable due to the curse of dimensionality resulting from the discretization of the space and time~\cite{oct}. A fast and locally approximate solution to the SOC problem is the Linear Quadratic Gaussian (LQG) case where the SOC problem is solved for the noise-free optimal trajectory and a local LQG model is constructed as perturbation around this trajectory. The local linear quadratic regulator computes a reasonable approximate solution to the original SOC problem if the model is close to the optimal noise-free trajectory.

The general duality between control and estimation~\cite{gd} and the notion of relating the cost and log-likelihood have motivated a new class of methods to approximately solve the SOC problem in a non-LQG setting. These methods are often referred to as \textit{control-as-inference} methods in literature which solve the SOC problem as an inference problem on a probabilistic graphical model (PGM). A PGM is a graphical model encoding complex relationships between random variables in the form of a graph. It is widely used in statistics and machine learning to model joint probability distributions of random variables. This graphical representation of probability distribution is advantageous as it allows the decomposition of the joint probability distribution as a product of factors by exploiting the structure of the model. Moreover, algorithms developed in this framework have shown propitious results in real-world applications (see e.g.~\cite{aico},~\cite{aicot},~\cite{convaico},~\cite{constraico},~\cite{i2c},~\cite{effsoc},~\cite{syed2023parameterized}).
A common limitation of the above  inference-based control approaches is the restriction to linear feedback controllers to achieve closed-form updates in a Gaussian setting. It is well known that nonlinear systems typically admit nonlinear optimal controllers, and hence the use of the existing linear controllers will yield sub-optimal performance in a nonlinear setting. In our prior work~\cite{syed2023parameterized}, we propose the Parameterized Input Inference for Control (\textit{PIIC}) algorithm where the controller is parameterized by a (possibly) nonlinear basis function of the state which allows formulating the unconstrained SOC problem as a parameter inference problem. 
Hence, one of the contributions of this paper is that we employ a barrier function approach to solve constrained SOC problems using the PIIC algorithm.

  In recent years, the design of structured controllers has received a lot of attention for applications in large-scale systems and multi-agent systems. A structured controller reduces the computational load by translating the topology of networked systems to the sparsity of the controller, facilitating distributed  controls at subsystems. An example of structured control is distributed optimal control for multi-agent systems, where the control of each agent contains information only from a subset of the agents. However, to the best of our knowledge, none of the existing inference-based control approaches have been developed in the structured control domain owing to the challenge of encoding and preserving the structure imposed on the control gain. Hence, the main contribution of this work is that we propose a \textit{structured}-PIIC algorithm to solve structured SOC problems in an inference-based control framework.

 The main contributions of this work are as follows:
   1)  We enhance the formulation of the PIIC algorithm~\cite{syed2023parameterized} to address constrained SOC problems, where the constraints include both state, control constraints and structural constraints on the state-feedback controllers. Although {structured} optimal control has been investigated for deterministic systems (see  e.g.,~\cite{Lin2011},~\cite{fardad2014},~\cite{jovanovic2016}), our approach provides an effective {structured} control solution for stochastic systems. The resulting algorithm is an instance of the EM procedure which has a guaranteed convergence to local optima. 
 2) We empirically demonstrate the effectiveness of the proposed algorithm with respect to constraint satisfaction and {structured} control using unicycle control problems. The algorithm outperforms the commonly-used Iterative Linear Quadratic Gaussian (ILQG) approach~\cite{ilqg} with reduced mean cost and cost variance.
  
The rest of the paper is organized as follows. Section~\ref{sec:ibc} reviews the formulation of the SOC problem in an inference-based control framework. Section~\ref{sec:sc} presents our algorithm to address constrained SOC problems.
Section~\ref{sec:simex} demonstrates the efficacy of our approach on 
a unicycle model {in constrained control and structured control scenarios}. 
 Section~\ref{sec:conclusion} concludes the paper.
 
\textit{Notation:} Let $\mathcal{N}(y|a,A)$ represent a random variable $y$ satisfying a Gaussian distribution in the normal form with mean $a \in \mathbb{R}^d$ and covariance $A \in \mathbb{R}^{d \times d}$ given by 
% \begin{align*}
    $\mathcal{N}(y|a,A) = \frac{1}{(2\pi)^{\frac{d}{2}}|A|^{\frac{1}{2}}}{\exp} \left({-\frac{1}{2}} (y - a)^\intercal A^{-1} (y - a) \right),$
    % \end{align*}
where $|A|$ represents the determinant of $A$. We use \text{blkdiag}$(A_1,A_2,\cdots,~A_n)$ to denote a block diagonal matrix with matrices $A_1,~A_2,~\cdots,~A_n$ on its principal diagonal. $\mathbb{I}_n$ denotes the identity matrix of size $n$. $\otimes$ denotes the Kronecker product. $\text{Tr}(\cdot)$ denotes the trace operator, and $\mathbb{E}(\cdot)$ denotes the expectation operator. $\mathbf{1}_{m\times n}$, $\mathbf{0}_{m\times n}$ denote the $m\times n$ matrices with entries $1$ and $0$, respectively. 

\section{Inference-based Stochastic Optimal Control}
\label{sec:ibc}
Consider a dynamical system given by
\begin{align}
    x_{t+1} &= F(\tau_t) + \eta_t,
    \label{eq:taumodel}
\end{align}
\useshortskip
where $\tau_t = [x^\intercal_t, u^\intercal_t]^\intercal \in \mathbb{R}^{n_x+n_u}$ is the state-control vector at time $t$,  $x_t \in \mathbb{R}^{n_x}$ and $u_t \in \mathbb{R}^{n_u}$ denote the state and control at time $t$, respectively. $F{: \mathbb{R}^{n_x}\times\mathbb{R}^{n_u}\rightarrow \mathbb{R}^{n_x}}$
is a nonlinear mapping of $x_t, u_t$, and $\eta_t \sim \mathcal{N}(\eta_t|0, \Sigma_{\eta_t})$ 
 represents additive Gaussian noise that models the uncertainty in the dynamics.
For a given finite-horizon $T,$ and a state-control sequence $[x_T, \tau_{0:T-1}]$, define the trajectory cost as 
  $ \mathcal{C}(x_T, \tau_{0:T-1}) = c_T(x_T) + \sum_{t=0}^{T-1} c_t(\tau_t),$
where $c_t: \mathbb{R}^{n_x+n_u} \rightarrow \mathbb{R}$ is a nonlinear mapping from the state-control space to the cost space for $t < T$ and $c_T: \mathbb{R}^{n_x} \rightarrow \mathbb{R}$ is a nonlinear mapping from the state space to the cost space at the terminal time $T$.
The considered SOC problem is given by
\vspace{0.3cm}
\begin{align}
    \underset{u_{0:T-1}}{\text{min}}&\mathbb{E}[\mathcal{C}(x_T, \tau_{0:T-1})]\label{eq:SOC}\\
    \text{such that } x_{t+1} &\sim \mathcal{N}(x_{t+1}|F(\tau_t), \Sigma_{\eta_t}),\nonumber\\
        \mathcal{K}(\tau_t) &> 0,\nonumber
\end{align}
where
$\mathcal{K}(\cdot) \in \mathbb{R}^{n_{\text{in}}}$ is such that $\mathcal{K}_j(\cdot):\mathbb{R}^{n_x} \times \mathbb{R}^{n_u} \rightarrow \mathbb{R}$, $j = 1,\cdots,n_{\text{in}}$ is a nonlinear mapping that defines an inequality constraint.
We assume that the feedback controller $u_t$ at each time step is parameterized by a (possibly nonlinear) basis function of the state, $\mathcal{B}_t(x_t) \in \mathbb{R}^{n_b}$, and unknown parameters $\Theta_t \in \mathbb{R}^{n_b \times n_u}$ such that
 \begin{align}
     p(u_t|x_t) = \mathcal{N}(u_t|\Theta_t^\intercal \mathcal{B}_t(x_t), \Sigma_{\delta_t}), 
     \label{eq:pcontr}
 \end{align}
 where $\delta_t$ represents a zero-mean Gaussian noise with covariance $\Sigma_{\delta_t}$ that models the uncertainty in control.

 The PGM for the SOC problem~\eqref{eq:SOC} is constructed with the state-control sequence as latent variables and the sequence of binary random variables $\mathcal{O}_{t} \in \{0, 1\}$, $t=0,\cdots, T$, as observed variables. The binary random variable $\mathcal{O}_t$ represents the notion of optimality or task fulfillment at each time step, i.e., $\mathcal{O}_t = 1$ when optimal state and action are observed at time $t$. Similar to the general duality between estimation and control~\cite{gd}, probabilistic inference approaches relate the probabilities to cost by assuming that the negative log-likelihood of observing the optimality/task fulfillment at time $t$ is proportional to the stage cost $c_t$, i.e.,
\begin{align}
    p(\mathcal{O}_t=1|\tau_t) &\propto \exp\{-c_t(\tau_t)\}.
    \label{eq:likelihood}
\end{align}
Hence, the likelihood of observing optimality at each time step is high if and only if the cost incurred is low.
We have shown in our prior work~\cite{syed2023parameterized} that the parameterization in~\eqref{eq:pcontr} yields nonlinear controllers for the unconstrained version of~\eqref{eq:SOC} using the EM procedure. The focus of this work is to extend the formulation to constrained and structured SOC problems.

 \section{Constrained Stochastic Optimal Control}\label{sec:sc}
We consider two types of constraints in the SOC problem. Section~\ref{sec:state} addresses inequality constraints on $\tau_t$, which are particularly useful for maintaining safety of the system and creating bounded controls. Section~\ref{sec:structure} examines structural constraints on the control, which can be used for designing distributed controllers. Corresponding examples are demonstrated in Section~\ref{sec:simex}. 
\subsection{State and control constraints}\label{sec:state}
We present an approach to embed inequality constraints on $\tau_t$ into the inference-based control formulation in Section~\ref{sec:ibc}. We are motivated by the barrier function method, which is a popular approach in optimization literature to solve a constrained optimization problem as a sequence of unconstrained optimization problems by adding a high cost for approaching the boundary of feasibility region from the interior~\cite[Chapter 5]{bertsekas2016}. It is also similar to the potential function approach commonly used for collision avoidance and motion planning \cite{Kavraki2008}. 
 
 Let the safe set for  constraint $j = \{1, \cdots, n_{\text{in}}\}$  be given by $\mathcal{C}_{s,j} = \{\tau_t \in \mathbb{R}^{n_x+n_u}| \mathcal{K}_j(\tau_t) > 0\},$
where $\mathcal{C}_{s,j}$ is assumed to be non-empty $\forall j$. A barrier function $B(\tau)$ is continuous in the interior of $\mathcal{C}_{s,j}$ and goes to $\infty$ as one of the constraints $\mathcal{K}_j$ approaches 0 from positive values. Motivated by this approach, we define a relaxed barrier function for each constraint, denoted by $c_{\text{in},j}(\tau_t)$, that evaluates to zero if and only if $\tau_t \in \mathcal{C}_{s,j}$, and is positive otherwise, i.e.,
 \begin{align}
    c_{\text{in},j}(\tau_t) &=  (\psi_j(\tau_t) )^\intercal Q^{\text{in}}_j  \psi_j(\tau_t) 
    \begin{cases}
        = 0, & \text{if $\tau_t \in \mathcal{C}_{s,j}$}\\
        > 0, & \text{otherwise,}
    \end{cases}
    \label{eq:constrcost}
\end{align}
where $\psi(\tau_t)$ is a (possibly) nonlinear function of $\tau_t$.
The $c_{\text{in},j}(\tau_t)$ can be considered the cost for the satisfaction of constraint $j$. It is positive when the constraint is violated and zero otherwise. As shown later, we employ a likelihood function $\exp(- c_{\text{in},j}(\tau_t))$ to encode the satisfaction of constraint $j$ into our inference-based control. According to~\eqref{eq:constrcost}, the likelihood function evaluates to $1$ in the safe set $\mathcal{C}_{s,j}$, which is the maximum of $\exp(- c_{\text{in},j}(\tau_t))$. Thus, satisfaction of constraint $j$ is encoded with a higher likelihood of occurrence. 

Let $\mathcal{O}^\tau_t$, 
$\mathcal{O}^{\text{in},j}_t$ denote the binary random variables corresponding to observing optimality in the cost, and in the satisfaction of constraint $j$, respectively. We prescribe
$p(\mathcal{O}_t=1|\tau_t) \propto p(\mathcal{O}^\tau_t=1|\tau_t) \prod_j p(\mathcal{O}^{\text{in},j}_t=1|\tau_t).$
Letting $p(\mathcal{O}^{\text{in},j}_t=1|\tau_t)\propto \exp(- c_{\text{in},j}(\tau_t))$, we rewrite 
\eqref{eq:likelihood} as 
\begin{align}\label{eq:total_likelihood}
p(\mathcal{O}_t=1|\tau_t) &\propto \exp\{-c_t(\tau_t) - \sum_{j=1}^{n_{\text{in}}} c_{\text{in},j}(\tau_t)\}.
\end{align}
Suppose that the trajectory cost $c_t(\tau_t)$ is quadratic. Adding the barrier function in~\eqref{eq:constrcost} as a cost to $c_t(\tau_t)$ yields  
$\forall~t=0,\cdots,T,$
\begin{align}
    &c_t(\tau_t) + \sum_{j=1}^{n_{\text{in}}} c_{\text{in},j}(\tau_t) 
    =  
    (x_t - x^d_t)^\intercal Q_t (x_t - x^d_t) 
   \nonumber\\
   &~~+ (u_t - u^d_t)^\intercal R_t (u_t - u^d_t)
    + \sum_{j = 1}^{n_{\text{in}}} (\psi_j(\tau_t))^\intercal Q^{\text{in}}_j  (\psi_j(\tau_t)),
    \label{eq:SOCconstr}
\end{align}
 where $Q_t \succeq 0, \ R_t \succeq 0,$ 
and $Q^{\text{in}}_{j} \succeq 0$, $j=1,\cdots,n_{\text{in}}$, are the  cost matrices. It then follows from~\eqref{eq:total_likelihood} and~\eqref{eq:SOCconstr} that
\begin{align}
    p(\mathcal{O}_t=1|\tau_t) &\propto\exp\{-\alpha(z^*_t - h(\tau_t))^\intercal \Gamma_t (z^*_t - h(\tau_t)))\}\nonumber\\
    &= \mathcal{N}(z_t = z_t^*|h(\tau_t), (\alpha\Gamma_t)^{-1}),
    \label{i2cobs}
\end{align}
where $\Gamma_t = \text{blkdiag}(
    Q_t,R_t,
    Q^{\text{in}}_1,\cdots,Q^{\text{in}}_{n_{\text{in}}})\in \mathbb{R}^{{n^*}\times {n^*}}$,
$h(\tau_t) = \bigg[
    \tau_t^\intercal
    ~~\psi(\mathcal{K}_1(\tau_t))~~\cdots~~\psi(\mathcal{K}_{n_{\text{in}}}(\tau_t))
\bigg]^\intercal  \in \mathbb{R}^{n^*},$
 $z_t^* = \begin{bmatrix}
    (\tau^d_t)^\intercal
    & 0&\cdots&0
\end{bmatrix}^\intercal  \in \mathbb{R}^{n^*}$ with ${\tau^d_t}^\intercal=[(x^d_t)^\intercal~(u^d_t)^\intercal]$, $n^* = (n_x+ n_u +n_{in})$, and $\alpha$ is the scale factor (hyperparameter) introduced to optimize the covariance of 
$\mathcal{O}_t$ to maximize the expected log-likelihood.

 An optimal trajectory is computed as the mean of the conditional or joint posterior distribution of the state-control trajectory given that the optimality
 is observed throughout the entire trajectory, i.e., $\mathcal{O}_{0:T} = 1$. The objective of the PIIC algorithm is to infer the parameters $\Theta_{0:T-1}$ and $\alpha$ that maximize the log-likelihood, i.e.,
\begin{align}
    \Theta_{0:T-1}^{*}, \alpha^* = \underset{\Theta_{0:T-1}, \alpha}{\text{argmax}} \log[ p(\mathcal{O}_{0:T}=1| \Theta_{0:T-1}, \alpha)].
    \label{eq:optobj}
\end{align}
The optimization problem in~\eqref{eq:optobj} is generally intractable. 
Thus, we resort to computing the parameters using {the} EM algorithm. {The EM algorithm is an iterative algorithm used to find maximum likelihood solutions for models with latent variables. It performs consecutive expectation (E-step) and maximization (M-step) steps in each iteration. The E-step computes the expected log-likelihood over the posterior distribution of latent variables and the consequent M-step computes the parameters that maximize this expectation. Each iteration of the EM algorithm results in a non-decreasing expected log-likelihood, thus guaranteeing convergence to a local maximum. We refer interested readers to~\cite{PRML} for a detailed introduction to the EM algorithm.}

Denote $\tau_{0:T-1}$ by $\tau$, $\mathcal{O}_{0:T}=1$ by $\mathcal{O}$, and $\Theta_{0:T-1}$ by $\mathbf{\Theta}$. Then the objective in~\eqref{eq:optobj} is rewritten as
\begin{align}   
\log[ p(\mathcal{O}| \mathbf{\Theta}, \alpha)] = \log \big[\int p(x_T, \tau,\mathcal{O}| \mathbf{\Theta}, \alpha)d\tau dx_T \big].
\label{eq:emobj}
\end{align}
The integrand in~\eqref{eq:emobj} is proportional to the joint posterior distribution given by
\begin{align}
    &p(x_T, \tau, \mathcal{O}, \mathbf{\Theta}, \alpha) = p(x_0) p(\mathcal{O}_T=1|x_T,\alpha) \nonumber\\
     &\qquad\prod_{t=0}^{T-1} p(x_{t+1}|\tau_t)p(\mathcal{O}_t=1|\tau_t, \alpha)p(u_t|x_t, \Theta_t).
    \label{eq:post}
\end{align}
Introducing $q(x_T,\tau)$, a known tractable distribution of $x_T$ and $\tau$,  we  obtain 
\begin{align*}
\log[ p(\mathcal{O}| \mathbf{\Theta}, \alpha)]=
    \log \bigg[\underset{q(x_T,\tau)}{\mathbb{E}}\left[\frac{p(x_T,\tau,\mathcal{O}|\mathbf{\Theta}, \alpha)}{q(x_T,\tau)}\right] \bigg]
    \end{align*}
    Using Jensen's inequality, we further get
    \begin{align}
    {\log[ p(\mathcal{O}| \mathbf{\Theta}, \alpha)]} &\geq\underset{q(x_T,\tau)}{\mathbb{E}} \log \bigg[\frac{p(x_T,\tau,\mathcal{O}|\mathbf{\Theta}, \alpha)}{q(x_T,\tau)} \bigg].
    \label{eq:jenineq}
\end{align}
Note that~\eqref{eq:jenineq} becomes equality for $q(x_T, \tau)= p(x_T,\tau|\mathcal{O})$. 
The PIIC algorithm optimizes the right-hand side of~\eqref{eq:jenineq} based on the EM procedure.
Hence, convergence to a local maximum is guaranteed~\cite{moon1996expectation}.

Substituting~\eqref{eq:post} in the M-step yields
 \begin{align}
     &\underset{\mathbf{\Theta}, \alpha}{\text{argmax}}\underset{q(x_T,\tau)}{\mathbb{E}}\bigg[\log p(x_0) + \sum_{t=1}^{T-1} \log p(x_{t+1}|\tau_t)+ \nonumber\\
     &\sum_{t=0}^T \log p(\mathcal{O}_t=1|\tau_t, \alpha) + \sum_{t=0}^{T-1}\log p(u_t|x_t, \Theta_t)\bigg].
      \label{eq:obj}
\end{align}
{To find {$\Theta_t^{k+1}$}, we take gradient of~\eqref{eq:obj} with respect to $\Theta_t$ and set it to zero, which  yields}
 \begin{align}
          {\Theta_t^{k+1} =  \left[\underset{q({\tau_t})}{\mathbb{E}}(\mathcal{B}_t(x_t)\mathcal{B}_t(x_t)^\intercal)\right]^{-1}\underset{ q({\tau_t})}{\mathbb{E}} (\mathcal{B}_t(x_t)u_t^\intercal).}
      \label{eq:piic}
 \end{align} 
It is straightforward to show that if the control parameter is time-invariant i.e., $\Theta_{0:T-1} = \Theta$ then 

\begin{align} 
\Theta^{k+1} = \big[\underset{q( \tau)}{\mathbb{E}}\big(\sum_{t=0}^{T-1}\mathcal{B}_t(x_t)(\mathcal{B}_t(x_t))^\intercal\big)\big]^{-1}\underset{q(\tau)}{\mathbb{E}} \big(\sum_{t=0}^{T-1} \mathcal{B}_t(x_t)u_t^\intercal\big).\label{eq:tipiic}\end{align}
Similarly,
 to find  $\alpha^{k+1}$ we take gradient of~\eqref{eq:obj} with respect to $\alpha$ and set it to zero, which yields
 \begin{align}
     & \alpha^{k+1} = \frac{(T-1)n_z + n_{z_T}}{\sum_{t=0}^\intercal \text{Tr}(\Gamma_t\underset{q({x_T}, \tau)}{\mathbb{E}}[(z_t^* - z_t)(z_t^* - z_t)^\intercal])},
     \label{eq:alpha}
 \end{align}
 where $q(x_T, \tau)= p(x_T,\tau|\mathcal{O})$.

In this paper, we define \textit{approximate inference} as the inference of the latent variables of a PGM. {Approximate inference} can also be defined as an approximation of the true posterior with a family of distributions that minimizes the KL divergence~\cite{aisoc1}.
Let $q_\pi(\tau) = \prod_{t=0}^{T-1} p(u_t|x_t) p(x_{t+1}|\tau_t),$
be the state-control distribution parameterized by $\Theta$ and 
    $q_{s}(\tau) = p(\tau|\mathcal{O})$
be the smoothed state-control distribution. 
\begin{prop}
The minimization of the KL divergence $KL(q_{s}|| q_\pi)$ is equivalent to the minimization of the objective~\eqref{eq:obj} with respect to the parameter $\Theta$. 
\end{prop}
\begin{proof}
From the definition, we have
    \begin{align}
        KL(q_{s}|| q_\pi) &= \int_\tau q_{s} \log(\frac{q_{s}}{q_\pi}) d\tau\nonumber\\
        &= \underset{q_{s}}{\mathbb{E}} \log(q_{s}) - \underset{q_{s}}{\mathbb{E}} \log(q_\pi).
    \label{eq:kld}
    \end{align}
    To minimize~\eqref{eq:kld} w.r.t. $\Theta$, we take the gradient and set it to zero, resulting in 
    \begin{align}
      \Theta^* =
      \left[\underset{q_{s}}{\mathbb{E}}[\sum_{t=0}^{T-1}\mathcal{B}_t(x_t)(\mathcal{B}_t(x_t))^\intercal]\right]^{-1}\left[\underset{q_{s}}{\mathbb{E}} [\sum_{t=0}^{T-1}\mathcal{B}^\intercal_t(x_t)u_t]\right],
      \label{eq:KLpiic}
    \end{align}
which is equivalent to \eqref{eq:tipiic}. \qed
\end{proof}

 \subsection{{Structured control}}\label{sec:structure}
{
Structured optimal control primarily deals with the design of static optimal controllers for interconnected systems with topological constraints. These topological constraints are translated as sparsity in the feedback gain. The problem of designing optimal controllers with structured feedback gains has been well studied for deterministic systems (e.g., see~\cite{Lin2011,fardad2014,jovanovic2016}). However, it has not been fully explored for stochastic systems.
We impose a structural constraint on the controller gain matrix $\Theta_t$.  
We assume that the state $x_t$ and the control $u_t$ are composed of $N$ subsystem states and $M$ subcontrols, respectively, i.e.,
$x_t=[(x^1_t)^\intercal,\cdots,  (x^N_t)^\intercal]^\intercal$ and $u_t=[(u^1_t)^\intercal,\cdots, (u^M_t)^\intercal]^\intercal$. The $x^i_t$ and $u^j_t$ can be multidimensional, $i=1,\cdots,N$, $j=1,\cdots,M$. Denote by $b^i(x^i_t) \in \mathbb{R}^{n_{b^i}}$ and by $\theta_t^{ij}$ the basis function  corresponding to $x_t^i$ and the submatrix of the controller gain $\Theta_t$ corresponding to $u^j$ and $b^i(x^i_t)$, respectively. The subcontrols $u^j_t$, $j=1,\cdots,M$, are parameterized as 
\begin{align}
    u^j_t ={\begin{bmatrix}
    (\theta_t^{1j})^\intercal&(\theta_t^{2j})^\intercal&\cdots&(\theta_t^{Nj})^\intercal
\end{bmatrix}} {\mathcal{B}_t(x_t)} +\delta^j_t,
\label{eq:subcontr}
\end{align}  
where $\mathcal{B}_t(x_t) = [
    (b^1(x^1_t))^\intercal~
    (b^2(x^2_t))^\intercal~
    \cdots~
(b^N(x^N_t))^\intercal]^\intercal$, $\delta^j_t \sim \mathcal{N}(\delta^j_t|0, ({\sigma^{j}_t})^2)$. We assume that $\delta^j_t$, $j=1,\cdots,M$, are i.i.d. zero mean Gaussian noise.
Following the notation in~\eqref{eq:pcontr}, we have 
$$\Theta_t = \begin{bmatrix}
    (\theta_t^{11})^\intercal&(\theta_t^{21})^\intercal&\cdots&(\theta_t^{N1})^\intercal\\
    \vdots&\vdots&\vdots&\vdots\\
    (\theta_t^{1M})^\intercal&(\theta_t^{2M})^\intercal&\cdots&(\theta_t^{NM})^\intercal
\end{bmatrix}^\intercal.$$   
Let $\mathcal{F}$ be the set of ordered pairs such that $(i,j)\in \mathcal{F}$ if the subcontrol $u^j_t$ can receive information from the subsystem state $x^i_t$. Consider the structured SOC problem:
\begin{align}
    \underset{u_{0:T-1}}{\text{min}}&\mathbb{E}[\mathcal{C}(x_T, \tau_{0:T-1})]    \label{eq:SOCstr}\\
    \text{such that } x_{t+1} &\sim \mathcal{N}(x_{t+1}|F(\tau_t), \Sigma_{\eta_t}),\nonumber\\
        \mathcal{K}(\tau_t) &> 0,~
        \theta^{ij}_{t} = \mathbf{0}_{n_{b^i} \times n_{u^j}}, \text{ if } (i,j) \notin \mathcal{F}.\nonumber
\end{align}
Our key idea to solve the structured SOC problem is to decompose the problem into multiple unstructured SOC problems in a lower dimensional subspace of nonzero entries corresponding to each element of $u_t$. Then, the inferred parameters are mapped back to the original vector space through an inverse transformation which 
preserves the structure during the inference procedure.

To capture the structural constraints, we define a structural identity (under element-wise matrix multiplication) of the feedback gain $\Theta_t$, denoted by $\varPhi \in \mathbb{R}^{n_b\times n_u}$. The $\varPhi$ is a block matrix whose $(i,j)^{\text{th}}$ block is all ones if $u^j_t$ depends on $b^i(x^i_t)$ and otherwise all zeros, that is, $\forall$ $i=1,\cdots,N$, $j=1,\cdots,M$,
\begin{align}
    \varPhi_{ij} = \begin{cases}
    \mathbf{1}_{n_{b^i} \times n_{u^j}}, & \text{if }(i,j)\in \mathcal{F}\\
    \mathbf{0}_{ n_{b^i}\times n_{u^j}}, & \text{otherwise.}
\end{cases}
\end{align}
Let $u_t^p$ be the $p^{\text{th}}$ element of $u_t$ and $\varPhi_p \in \mathbb{R}^{n_b}$ be the  $p^{\text{th}}$ column of $\varPhi$, where $p  = \{1,\cdots,n_u\}$. For every $\varPhi_p$, there exists an $\mathcal{S}_p: \mathbb{R}^{n_b} \rightarrow \mathbb{R}^{\tilde{n}_b}$ that maps $\varPhi_p$ to {its} lower dimensional  non-zero entries $\tilde{\varPhi}_p \in \mathbb{R}^{\tilde{n}_b}$, where $\tilde{n}_b \leq n_b$. Hence, $\mathcal{S}_p $ can be applied to the $p^{\text{th}}$ column of $\Theta_t$, denoted by $\Theta^p_t$, to extract its non-zero entries, denoted by $\tilde{\Theta}_t^p \in \mathbb{R}^{\tilde{n}_b}$, i.e., $\tilde{\Theta}_t^p=\mathcal{S}_p{\Theta}_t^p$. Similarly, we let $\tilde{\mathcal{B}}^p_t({x}_t)=\mathcal{S}_p{\mathcal{B}}^p_t({x}_t)$. Also, for every $\mathcal{S}_p$, there exists an $\mathcal{S}_p':\mathbb{R}^{\tilde{n}_b} \rightarrow \mathbb{R}^{n_b}$ that maps $\tilde{\Theta}_t^p$ back to ${\Theta}_t^p$.

For example, consider an interconnected system with four subsystem states and three subcontrols, i.e., $x_t = [x^1_t~x^2_t~x^3_t~x^4_t]^\intercal\in\mathbb{R}^4$ and $u_t = [
    u^1_t~u^2_t~u^3_t
]^\intercal\in\mathbb{R}^3$. Assume that $u^i_t$'s are linear functions of the states. Consider the following structural constraints on $u^i_t$'s: $u^1_t$ depends only on $x^1_t$ and $x^3_t$, $u^2_t$ only on $x^1_t$, $x^2_t$, and $x^4_t$, and $u^3_t$  only on $x^3_t$ and $x^4_t$. Then~\eqref{eq:subcontr} takes the form $${u_t} = \underbrace{\begin{bmatrix}    \theta^{11}&0&\theta^{31}&0\\   \theta^{12}&\theta^{22}&0&\theta^{42}\\    0&0&\theta^{33}&\theta^{43}\\
\end{bmatrix}}_{\Theta^\intercal_t} \underbrace{\mathcal{B}_t(x_t)}_{x_t} + \delta_t,$$ where $\delta_t = [\delta^1_t~\delta^2_t~ \delta^3_t]^\intercal$. By definition, $\varPhi = \begin{bmatrix}
    1&0&1&0\\
    1&1&0&1\\
    0&0&1&1\\
\end{bmatrix}^\intercal.$  Then,  $\varPhi_1=\begin{bmatrix}
    1& 0& 1& 0
\end{bmatrix}^\intercal,$ $\mathcal{S}_1= \begin{bmatrix}
    1& 0& 0& 0\\
    0& 0& 1& 0
\end{bmatrix}$ and $\mathcal{S}'_1 = \mathcal{S}_1^\intercal$.
Therefore, $\tilde{\Theta}_t^1 =  \mathcal{S}_1 \Theta_t^1=[\theta^{11}~\theta^{13}]^\intercal$ and $\Theta_t^1 = \mathcal{S}'_1 \tilde{\Theta}_t^1$. 
Similarly, $\mathcal{S}_2, ~\mathcal{S}_3$ can be computed corresponding to $\Theta_t^2,~\Theta_t^3$ respectively.

Using the notation $\tilde{\Theta}^p_t$ and $\tilde{\mathcal{B}}^p_t({x}_t)$, taking the gradient of~\eqref{eq:obj} against $\tilde{\Theta}^p_t$, and equating it to zero yields the update equation for $\tilde{\Theta}^p_t$  as
  \begin{align}
     (\tilde{\Theta}^p_t)^{k+1}
      = \left[\underset{q(\tau_t)}{\mathbb{E}} [\tilde{\mathcal{B}}^p_t(x_t)(\tilde{\mathcal{B}}^p_t(x_t))^\intercal]\right]^{-1}\underset{q(\tau_t)}{\mathbb{E}} [\tilde{\mathcal{B}}^p_t(x_t)(u_t^p)^\intercal].
       \label{eq:piicstr}
 \end{align}
From an implementation perspective, a time-invariant control parameter $\Theta$ may be advantageous. Following a similar approach to~\eqref{eq:piicstr}, we obtain the time-invariant parameter update as 
\begin{align}
     &(\tilde{\Theta}^p)^{k+1} = \nonumber\\
     &\big[\underset{q(\tau)}{\mathbb{E}} \big[\sum_{t=0}^{T-1}\tilde{\mathcal{B}}^p_t(x_t)(\tilde{\mathcal{B}}^p_t(x_t))^\intercal\big]\big]^{-1}\underset{q(\tau)}{\mathbb{E}} \big[\sum_{t=0}^{T-1}\tilde{\mathcal{B}}^p_t(x_t)(u_t^p)^\intercal\big].
     \label{eq:piicstrti}
\end{align}
The covariance of the controller $\sigma^{p}$ is updated $\forall$~$p = {1,\cdots,n_u}$ and $t = {0,\cdots,T-1}$ using 
\begin{align}
\sigma_t^{p} =\underset{q(\tau)}{\mathbb{E}}(u_t^p -(\tilde{\Theta}^p_t)^\intercal\tilde{\mathcal{B}}_t({x}^p_t))(u_t^p -(\tilde{\Theta}^p_t)^\intercal\tilde{\mathcal{B}}_t({x}^p_t))^\intercal.
    \label{eq:sdt}
\end{align}
Algorithm~\ref{alg:gpiic} below  summarizes the structured parameterized input inference for control (structured PIIC) algorithm. It performs the E-step and the M-step iteratively until convergence. 
The structure imposed on the control parameter $\Theta$ is preserved by performing updates on the non-zero subsets of each subsystem using~\eqref{eq:piicstr} and~\eqref{eq:sdt}. In our implementation, we claim convergence of the algorithm if the infinity norm of the difference between the state trajectories in two consecutive iterations is less than a threshold.
\begin{algorithm}[htpb]
\caption{Structured {PIIC} algorithm}\label{alg:gpiic}
\begin{algorithmic}
     {\REPEAT
     \STATE {\text{E-step: Compute}
     \begin{align*}
     q^{k+1}&= p({x_T},\tau|\mathcal{O},\mathbf{\Theta}^k,\alpha^k)\\
      Q(\mathbf{\Theta},\alpha|\mathbf{\Theta}^k,\alpha^k)&=  \underset{({x_T},\tau)\sim q^{k+1}}{\mathbb{E}}\log [p(x_T, \tau,\mathcal{O}|\mathbf{\Theta}, \alpha)]
     \end{align*}    
      \text{M-step:}}
      \FOR{t= $0$ : $T-1$}{\FOR {p = $1$ : $n_u$}{
      \STATE{\text{Update $\tilde{\Theta}_t^p,$ $\sigma_t^p$ using \eqref{eq:piicstr},~\eqref{eq:sdt}, respectively.}}
        \STATE  $\Theta_t^p = \mathcal{S}_p'(\tilde{\Theta}_t^p)$        
          }
      \ENDFOR
      \STATE{\text{Update $\Theta_t = \begin{bmatrix}
      \Theta_t^1&\cdots&\Theta_t^{n_u}
      \end{bmatrix}^\intercal$}}
      \STATE{Update \text{$\Sigma_{\delta_t} = \text{blkdiag}(\sigma_t^1,~\cdots,~\sigma_t^N)$}}}
      \ENDFOR
      \STATE{Update $\alpha$ using \eqref{eq:alpha}}
      \UNTIL{convergence}
 }\end{algorithmic}
 
\end{algorithm}
}
As shown in the appendix, Algorithm~\ref{alg:gpiic} recovers the Gaussian I2C~\cite{effsoc} for linear dynamics without any constraints if $\mathcal{B}_t(x_t) = [x^\intercal_t~
        1
    ]^\intercal$ and $\Theta_t$ does not have a specific structure. 
\section{Simulation examples}
\label{sec:simex}

In this section, we demonstrate the effectiveness of the PIIC algorithm for inference of constrained and structured stochastic optimal controllers. In Section~\ref{sec:obsavoidance}, we demonstrate the effectiveness of the barrier function approach for a unicycle obstacle avoidance problem. We also study the performance of the PIIC for the choice of two smoothing approaches and compare them with the ILQG baseline. In Section~\ref{sec:formncontr}, we illustrate the utility of the PIIC for distributed formation control of four unicycle robots. The common simulation parameters are 
step size $dt$ = $0.05$, and the state cost matrix $Q_t$ =  $\mathbb{I}_3$. In Section~\ref{sec:quadsim}, we consider the problem of navigating a quadcopter to a desired position in an obstacle-filled windy environment using the PIIC algorithm. 

\subsection{Obstacle avoidance}\label{sec:obsavoidance}
Consider a unicycle robot whose dynamics are given as
\begin{align}
X_{t+1} &=  X_t + dt f(X_t, u_t) +\eta_t,\label{eq:unicycle}
\end{align}
where at any time instant $t$, $X_t = [x_t~y_t~\theta_t]^\intercal \in \mathbb{R}^3$ denotes the 2-dimensional positions and heading of the robot, $u_t = [v_t~ \omega_t ]^\intercal \in \mathbb{R}^2$ denotes the linear and angular velocities of the robot, $f_t(X_t,u_t) = [
 v_t \cos(\theta_t)~
v_t \sin(\theta_t)~ \omega_t]^\intercal$ denotes the nonlinear unicycle dynamics, $\eta_t \sim \mathcal{N}(\eta_t|0, \Sigma_{\eta_t})$ corresponds to the process noise, and $dt$ denotes the step size for discretization.
We consider the controller parameterization of the form~\eqref{eq:pcontr} where $\mathcal{B}(X_t) = [
    x_t~y_t~\theta_t~1
]^\intercal.$ 

The goal of the SOC problem is for the unicycle to navigate to a desired position without collision with obstacles. Let $\mathcal{A}$ be the set of obstacles. For $j\in\mathcal{A}$, we define \begin{align*}
    \mathcal{K}_j(\tau_t) &= [(x_t - x_{obs,j})^2 + (y_t - y_{obs,j})^2 - (r_{obs,j}+ r_s)^2],%\\
    \end{align*}
where  $(x_{obs,j}, y_{obs,j})$ and $r_{obs,j}$ are the center and the radius of the $j^{\text{th}}$ obstacle, respectively, and {$r_s$ denotes its safety radius}.
Let the unsafe set for the robot be
$\mathcal{C}_u = \{(x,y)\in \mathbb{R}^2| \mathcal{K}_j(x,y) < 0, ~\forall~ j \in \mathcal{A}\}$. Then, the safe set for the collision avoidance constraint $\mathcal{C}_s = {\mathbb{R}^2\backslash \mathcal{C}_u}$. In our simulations, we choose $\psi_j(\tau_t)$ in~\eqref{eq:constrcost} as
\begin{align}
    \psi_j(\tau_t) &= \begin{cases}
        0, & \text{if $\tau_t \in \mathcal{C}_s$}\\
        \gamma( 1 - \tanh(\epsilon \mathcal{K}_j(\tau_t))), &\text{otherwise,}
    \end{cases}
    \label{eq:tanh}
\end{align}
where $\gamma, \epsilon \in \mathbb{R}_+$ are tunable parameters to vary the tightness of the constraint and smoothness of $\psi_j(\cdot)$, respectively. 
For simulations, we choose $\gamma$= $1$ and $\epsilon$= $1$ in the barrier function~\eqref{eq:tanh}.
Other simulation parameters are given in Table~\ref{table:unisim}. 
Fig.~\ref{fig:probtanh} shows the variation of $p(\mathcal{O}^{\text{in},j}_t=1|\tau_t)$ with respect to $\gamma$  in~\eqref{eq:tanh}. 
We see that as $\gamma$ increases, the constraint becomes more conservative, resulting in a lower likelihood of constraint violation. 
\begin{figure}[htpb]
    \centering        \includegraphics[width=0.6\textwidth]{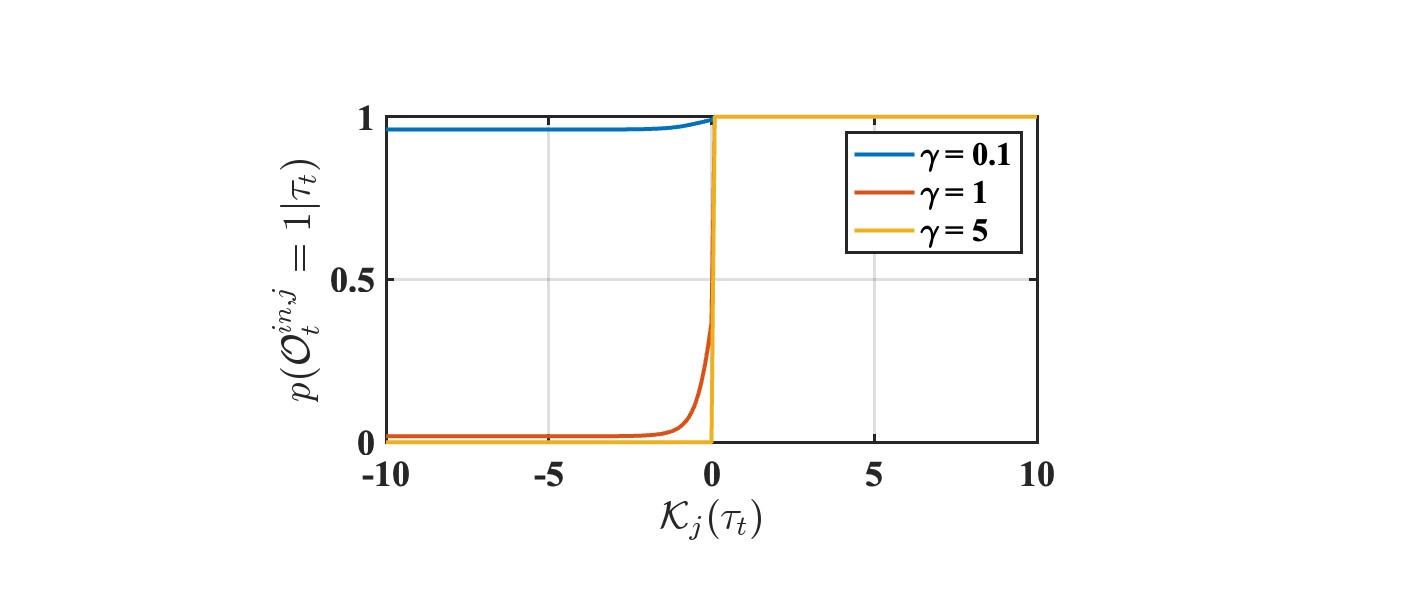} 
\caption{{Variation of $p(\mathcal{O}^{\text{in},j}_t=1|\tau_t)$ for $Q^{in}_j = 1$ and different values of $\gamma$.}} 
    \label{fig:probtanh}
\end{figure}
\begin{table*}[htpb]
\caption{{Simulation parameters for the unicycle example.}}
\label{table:unisim}
\centering
  \scalebox{1}{
\begin{tabular}{|l|l|}
    \hline
    Simulation parameters & Value\\
      \hline
      Process noise covariance, $\Sigma_{\eta_t}$ &  diag($10^{-3},10^{-3},10^{-3}$)\\
      Cost matrices, 
      $\{Q_{t,\text{obs}},Q_T,R_t\}$& $\{20,10~ \mathbb{I}_3,0.5~\mathbb{I}_2\}$\\
    \hline
  \end{tabular}}
\end{table*}

We investigate the effect of the choice of smoothing algorithm on the overall performance of the PIIC algorithm. We employ unscented smoothing (UPIIC), and \textit{factor graph optimization} (FGPIIC) in the E-step of Algorithm~\ref{alg:gpiic}. 
The \textit{unscented smoothing} is analogous to the unscented Kalman smoothing~\cite{sarkka2008} 
except that we utilize it to compute the smoothed \textit{state-control} distribution rather than just the state distribution.
Factor graph optimization solves the smoothing problem as a nonlinear least squares problem. 
This is possible due to the fact that the maximum-a-posteriori (MAP) inference on a nonlinear factor graph with Gaussian noise models is equivalent to nonlinear least squares problem~\cite{fgforrp}.
We interface with the GTSAM library~\cite{gtsam} to implement the factor graph generation and optimization. This approach is well known to be computationally efficient. 

We compare the performance of the UPIIC, FGPIIC with the ILQG algorithm~\cite{ilqg}. The ILQG algorithm does not accommodate state constraints. Hence, we use the modified cost~\eqref{eq:SOCconstr} with the barrier function candidate~\eqref{eq:tanh} to impose the obstacle avoidance constraint for a fair comparison. Fig.~\ref{fig:smoothingcomp} shows the trajectories for 50 MC simulations with the corresponding covariance ellipses for $T = 200$.
\begin{figure}[htpb]
    \centering
    \includegraphics[width = 0.6\textwidth]{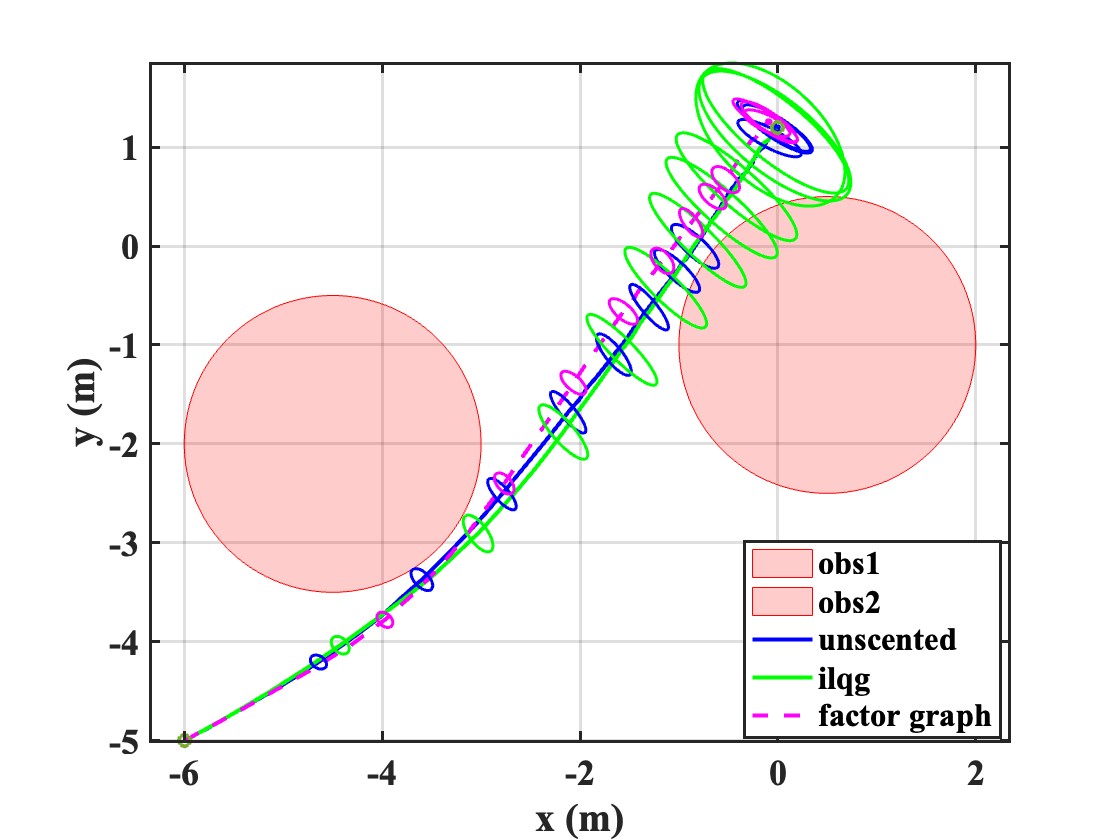}
    \caption{{Comparison of the trajectories with the feedback controllers inferred using ILQG, PIIC with unscented smoothing, and factor graph optimization.} }
    \label{fig:smoothingcomp}
\end{figure}

The mean and standard deviation of the incurred trajectory cost are shown in Table~\ref{table:smcomptable}.
We observe that the FGPIIC has the superior performance followed by UPIIC and ILQG. This can be attributed to the fact that each iteration of the FGPIIC performs multiple iterations of factor graph optimization until a level of convergence is reached whereas the UPIIC performs only one pass of the smoothing step per iteration, yielding in sub-optimal trajectories compared to FGPIIC. We also observe that the ILQG approach suffers from poor convergence, leading to higher variance in the trajectories and a greater number of constraint violations. We have repeated the same comparison for a target reaching problem without obstacles and the resulting trend was similar. 
\begin{table}[htpb]
\caption{{Comparison of the average cost and standard deviation for 50 MC simulations with the feedback controllers inferred using ILQG, UPIIC, and FGPIIC for unicycle target reaching example with and without obstacles.}\\}
\label{table:smcomptable}
\centering
 \scalebox{1} {
\begin{tabular}{|l|l|l|}
    \hline
    & {Without obstacles} & 
      {With obstacles} \\
      \hline
    UPIIC & 
    64.12 $\pm$ 4.26 &
    92.26 $\pm$ 31.39
    \\
     FGPIIC &
     \textbf{58.69 $\pm $ 3.18}&
    \textbf{61.89 $\pm$ 4.61}
    \\
         ILQG & 
         79.16 $\pm$ 19.02&
        189.02 $\pm$ 196.68
        \\
    \hline
  \end{tabular}}
\end{table} 

We also perform an empirical study on the effect of $\gamma$ in the  barrier function~\eqref{eq:tanh} on the inferred controller and the trajectory of the unicycle robot. We restrict to a single obstacle to visualize a more pronounced effect. Fig.~\ref{fig:gammacomp} shows the two trajectories resulting from  controllers inferred using different values of $\gamma$. We observe that for higher values of $\gamma$, the minimum distance of the trajectory from the obstacle increases, i.e., the controller becomes more conservative. 
\begin{figure}[htpb]
    \centering
    \includegraphics[width = 0.6\textwidth]{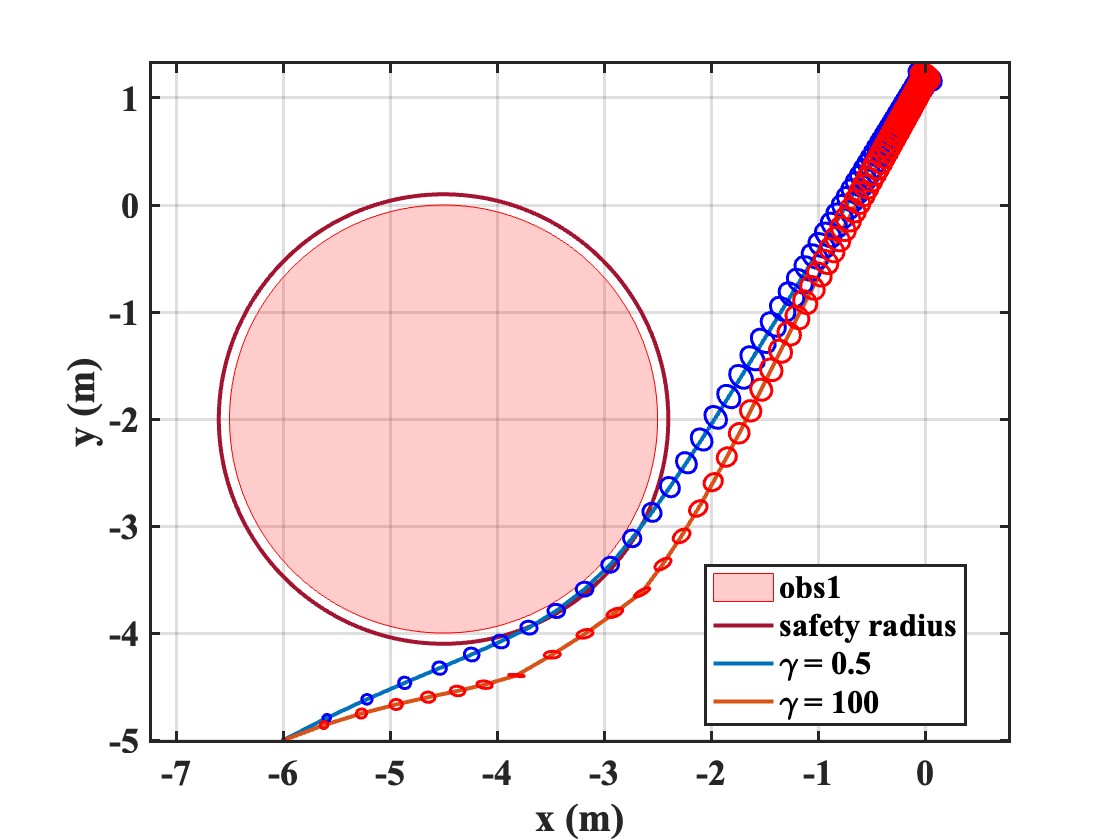}
    \caption{{Comparison of the trajectories of unicycle model with covariance ellipses for $\gamma = 0.5$ and $\gamma = 100$.}}
    \label{fig:gammacomp}
\end{figure}
\vspace{0.00mm} 
Due to the presence of process noise there is finite probability of constraint violation. However, for higher values of $\gamma$, the conservatism of the controller yields less constraint violations, i.e., the deviations from the inferred trajectory exist but remain in the safe set $\mathcal{C}_s$,  resulting in satisfaction of the actual constraint with a higher probability. We corroborate the claim in Table~\ref{table:cvgamma}, which shows that the number of constraint violations decreases as $\gamma$ increases.
\begin{table}[htpb]
\caption{{Comparison of the number of constraint violations occurred in 50 MC simulations for various values of $\gamma$.}\\}
\label{table:cvgamma}
\centering
 \scalebox{1}{
\begin{tabular}{|c|l|l|l|l|}
    \hline
     $\gamma$ & 1 & 2 & 5 & $\geq 10$\\
     \hline
      $\#$ of constraint violations &
    38
     &56
        & 11
        & 0\\
    \hline
  \end{tabular}}
\end{table}

\subsection{Formation Control}
\label{sec:formncontr}
We consider formation control of four unicycle robots modeled as~\eqref{eq:unicycle}.
The objective is to find a stochastic optimal controller that navigates to desired goal positions with minimal control energy applied by each agent while closely maintaining a desired square formation and avoiding collision with obstacles. 
We define the individual cost of robot $i$ as $\mathcal{C}^i_{x,al}(\tau^i_t) = (X^i_t - X^i_d)^\intercal Q^i_t (X^i_t - X^i_d) + 
\sum_{j=1}^{n_{\text{in}}}\psi^\intercal_j(\tau^i_t) ~Q^{i, \text{in}}_j  ~\psi_j(\tau^i_t)$, $\mathcal{C}^i_{u,al}(\tau^i_t) = (u^i_t)^\intercal R^i_t (u^i_t)$,
where at time $t$, $X^i_t$ and $X^i_d$ denote the state and the desired state of agent $i$, respectively, $u^i_t$ denotes the control input of agent $i$, and $\psi_j(\cdot)$ is the barrier function as in~\eqref{eq:tanh}. 
Let $X_t = [
    (X^1_t)^\intercal~(X_t^2)^\intercal~\cdots~(X_t^N)^\intercal
]^\intercal$ $ \in \mathbb{R}^{Nn_x}$ be the state  of all the agents $i \in \mathcal{V}$ in the formation. We assume a linear controller for each agent of the form $\mathbb{E}(u^i_t) = K_t^i X_t + k_t^i.$ Let $B \in \mathbb{R}^{N \times M}$ denote the incidence matrix of an undirected graph $\mathcal{G} = \{\mathcal{V}, \mathcal{E}\}$ corresponding to the formation, where $M$ is the cardinality of the edge set $\mathcal{E}$. Define the  formation cost as $\mathcal{C}_{nl}(\tau_t) =((B \otimes \mathbb{I}_{n_x})^\intercal X_t - \delta_*)^\intercal Q_{f} ((B \otimes \mathbb{I}_{n_x})^\intercal X_t - \delta_*),$
where 
$\delta_* = [
    {\delta^1_*}^\intercal~{\delta^2_*}^\intercal~\cdots~{\delta^M_*}^\intercal
]^\intercal \in \mathbb{R}^{Mn_x}$ represents the vector of formation targets along each edge $e \in \mathcal{E}$, and $Q_{f}$ is a positive semi-definite block diagonal cost matrix. For the unit square formation in the simulation, $\mathcal{V}= \{1,2,3,4\},$ $\mathcal{E} = \{(2,1),(4,1),(2,3),(4,3)\},$
and $\delta_* = [
    0~ 1~ 0~ -1~ 0~ 0~ 0~ 1~ 0~ 1~ 0~ 0
]^\intercal$.
The total trajectory cost for the optimal formation control problem is given by $\mathcal{C}(x_{0:T},u_{0:T-1}) =\sum_{t=0}^T \mathcal{C}_{nl}(\tau_t) +\sum_{i=1}^4 [ \sum_{t=0}^{T} \mathcal{C}^i_{x,al} (\tau^i_t)+ \sum_{t=0}^{T-1} \mathcal{C}^i_{u,al}(\tau^i_t)],$
where $T$ is the time horizon  set to $100$. Additional simulation parameters are given in Table~\ref{table:frmsim}.
\begin{table}[htpb]
\caption{{Simulation parameters for the formation control example.}}
\label{table:frmsim}
\centering
 \scalebox{1}{
\begin{tabular}{|l|l|}
    \hline
    Simulation parameters & Value\\
      \hline
      Process noise covariance, $\Sigma^i_{\eta^i_t}$ &  diag($10^{-3},10^{-3},10^{-4}$)\\
      Linear velocity limits, $v^i_t$ & $[0,8]~m/s$\\
      Angular velocity limits, $\omega^i_t$&$[-1.5,1.5]~\text{rad}/s$\\
    $\{Q^i_t,Q^i_{t,\text{obs}},Q^i_{t,\text{lim}},Q^i_T,R^i_t,Q_{t,f}\}$ & $\{\mathbb{I}_3,50~\mathbb{I}_4,50~\mathbb{I}_4,50~ \mathbb{I}_3,\mathbb{I}_2,50~ \mathbb{I}_{12}\}$\\
    \hline
  \end{tabular}}
\end{table}

We impose a \textit{3-agent partially decentralized} structure on the controller wherein each agent has access to the state information of itself and two other agents in the formation. Figure~\ref{fig:frmn} shows the formation trajectory of the robots with covariance ellipses using {Algorithm~\ref{alg:gpiic}}. We observe that the agents reach close to their target positions while avoiding the obstacles and respecting the square formation as closely as possible. 
\begin{figure}[htpb]
    \centering
    \includegraphics[width = 0.6\textwidth]{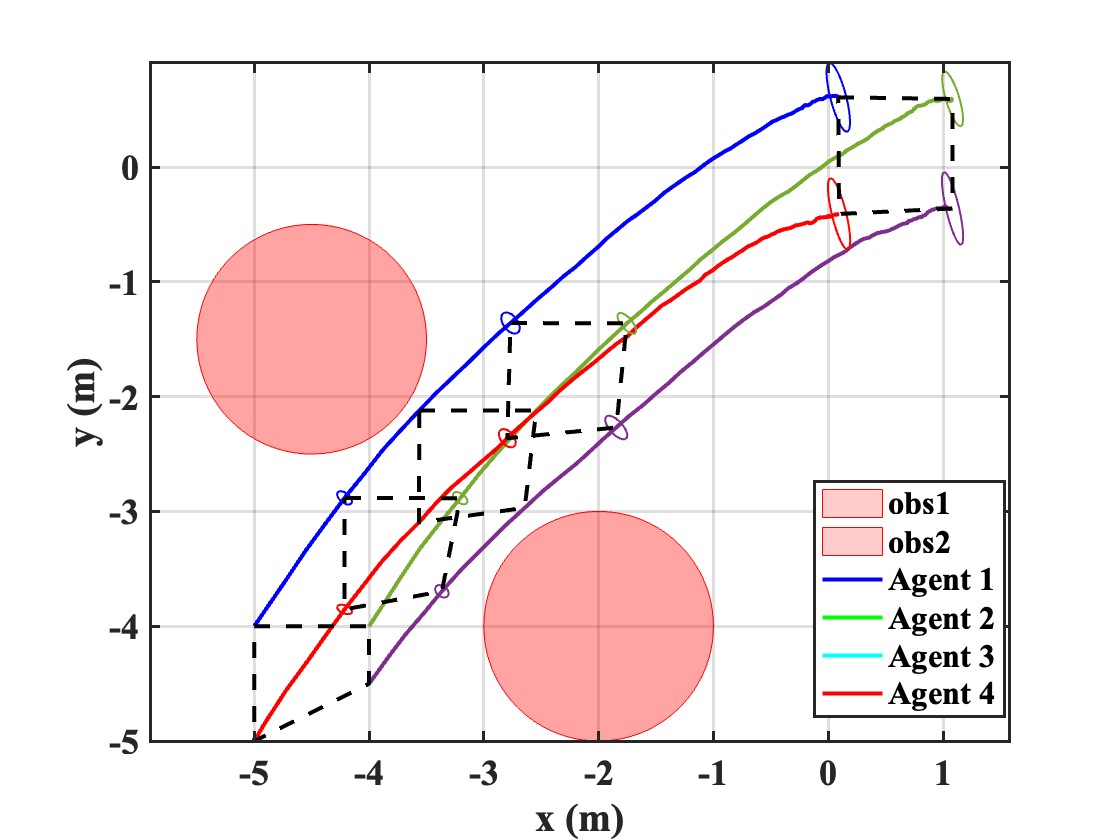}
    \caption{{Snapshots of X-Y trajectory of the unicycle formation with corresponding covariance ellipses.} }
    \label{fig:frmn}
\end{figure}

We next investigate this problem under three additional controller structures. A \textit{centralized} structure is where each agents has access to the state information of all the agents in the formation, a \textit{2-agent partially decentralized} structure is where each agent has access to the state information of itself and the agent diagonally opposite to it in the formation, and a \textit{decentralized} structure is where each agent has access to only its own state information. Table~\ref{table:frmcomptable} shows the average cost and standard deviation for 50 MC simulations. The centralized structure incurs the least average cost owing to its full information of the global state of the agents. It is followed by the 3-agent and 2-agent partially decentralized structures, respectively. The decentralized structure incurs the highest average cost. The increase of the  cost is correlated to the decrease in the information available to each agent, yielding controllers with degrading performance.
\begin{table}[htpb]
\caption{{Comparison of the average cost and standard deviation for 50 MC simulations with 
different controller structures for the 4-unicycle formation control example.\\}}
\label{table:frmcomptable}
\centering
  \scalebox{1}{
\begin{tabular}{|l|l|}
    \hline
    Controller structure & Average Cost\\
      \hline
    Centralized & 
    \textbf{664.67 $\pm$ 120.61 }
    \\
     Partially decentralized (3-agent) &
     \textbf{713.21 $\pm$ 116.82}
    \\
    Partially decentralized (2-agent) &
     728.88 $\pm$ 128.41\\
        Decentralized & 
         749.22 $\pm$ 135.88\\
    \hline
  \end{tabular}}
\end{table} 

\subsection{Quadcopter}\label{sec:quadsim}

% This problem is studied both with and without external wind.

To incorporate the wind in quadcopter dynamics, we use the formulation in~\cite{iekfquad} except that  rotation is represented using roll-pitch-yaw angles instead of quaternions. {The angular rates of rotation and the thrust are the control inputs.}
% \begin{align}
%     \dot{d} &= Ad\nonumber\\
%     v_w &= Cd,
%     \label{eq:winddyn}
% \end{align}
% where $d \in \mathbb{R}^{n_d \times 1}$, $A \in \mathbb{R}^{n_d \times n_d}$, and $C \in \mathbb{R}^{3 \times n_d}$. The wind velocity modeled as~\eqref{eq:winddyn} can be used to express both stationary and non-stationary wind profiles. Then, the quadcopter dynamics with external wind is given by
\begin{align}
     \dot{x}&= R_{\psi} R_{\theta} R_{\phi} v_r + Cd\nonumber\\
    \dot{v}_r &= v_r \times \omega + R_{\phi}^\top R_{\theta}^\top R_{\psi}^\top \bm{g} + \frac{1}{m} \bm{f} + \frac{1}{m} \bm{f}_{drag} \nonumber \\
    \begin{bmatrix}
         \dot{\phi}\\
      \dot{\theta}\\
       \dot{\psi}
    \end{bmatrix} &= \begin{bmatrix}
        1&0&-\sin(\theta)\\
       0&\cos(\phi)&\sin(\phi)\cos(\theta)\\
       0&-\sin(\phi)&\cos(\phi)\cos(\theta)
     \end{bmatrix}^{-1} \omega\\
\dot{d} &= Ad,  \label{eq:quadwinddyn}
\end{align}
where $v_r$ is the relative air velocity in body frame. Similar to~\cite{iekfquad}, the drag force $\bm{F}_{drag}$ is modeled in terms of $v_r$ as
% \begin{align}
    $\bm{f}_{drag} = \frac{1}{2} \rho D |v_r|v_r$,
% \end{align}
where $\rho$ is the air density, $D$ is the diagonal coefficient matrix. We define the state $X = \begin{bmatrix}
    x&v&\phi&\theta&\psi&d
\end{bmatrix}^\top \in \mathbb{R}^{12 \times 1}$, control $U = \begin{bmatrix}
    \omega_x&\omega_y&\omega_z&F
\end{bmatrix}^\top \in \mathbb{R}^{4 \times 1}$, and the dynamics $F_t(X, U)$ given in~\eqref{eq:quadwinddyn} for the quadcopter simulations with external wind. The SOC problem was solved using Algorithm~\ref{alg:gpiic} with a linear basis controller $\mathcal{B}^{LWB} = \begin{bmatrix}
    x&v&\phi&\theta&\psi&d&1
\end{bmatrix}^\top$.
The parameters used in this simulation are summarized in Table~\ref{table:quadsim}.

\begin{table}[htpb]
\caption{ {Simulation parameters for the quadcopter example.}}
\label{table:quadsim}
\centering
   {
\begin{tabular}{|l|l|}
    \hline
    Simulation parameters & Value\\
      \hline
      Step size, $dt$ & $0.02$\\
      Horizon, $T$&$650$ steps\\
      Initial state mean, $\mu_{x_0} $& $[20~20~0~0~0~0~0~0~0~2~-2~0]^\top$\\
      Initial state covariance, $\Sigma_{x_0}$ & $10^{-4} \mathbb{I}_{12}$\\
      Initial control mean, $\mu_{U_0}$ & $[0~~0~~0~~\bm{g}]^\top$\\
      Controller covariance, $\Sigma_{\delta_t}$ & $10^5 \mathbb{I}_{4}$\\
      Goal state, $z_T$ & $[35~33~7~0~0~0~0~0~0~2~-2~0]^\top $\\
      Process noise covariance, $\Sigma_{\eta_t}$ &  diag($
0,0,0,10^{-4},10^{-4},10^{-4},0,0,0,0,0,0$)\\
      Barrier function parameter, $\gamma$& $0.5$\\
      Barrier function parameter $\epsilon$& $1$\\
      State cost matrix, $Q_t$ &  $\mathbb{I}_{12}$\\
      Avoidance cost matrix, 
      $Q_{t,\text{obs}}$& $1000$\\
      Terminal cost matrix, $Q_T$ &  diag($100,100,100,10,10,10,10,10,10,10,10,10$)\\
      Control cost matrix, $R_t$ &  diag($0.1,0.1,0.1,0.1$)\\
    \hline
  \end{tabular}}
\end{table}

 The key advantage of the PIIC algorithm is the ability to admit nonlinear optimal controllers for nonlinear systems. Hence, we define a nonlinear obstacle-aware basis by appending a linear basis  $\mathcal{B}^{LWB} = [x~~v~~\phi~~\theta~~\psi~~d~~1]^\top$ with the sum of distances of the quadcopter to the boundaries of all the obstacles in the environment. It is denoted by $\mathcal{B}^{OA} = [x~~v~~\phi~~\theta~~\psi~~d~~1~~\sum_{i=1}^{n_o} c_i]^\top$ where $c_i = (x - x_{obs,i})^2 + (y - y_{obs,i})^2 - (r_{obs,i})^2 \ \forall i$.

%  \begin{figure}
%     \centering
%     % \begin{subfigure}{0.46\textwidth}
%     \includegraphics[width=0.46\textwidth]{figures/quadcopter/_2.png}
%     % \end{subfigure}
%     % \hfill
%     % \begin{subfigure}{0.46\textwidth}
%     \includegraphics[width=0.46\textwidth]{figures/quadcopter/_2-1.png}
%     % \end{subfigure}
%     \caption{Position trajectory of the quadcopter example with external wind.}
%     \label{ref:trajoawind}
% \end{figure}
% \begin{figure}
% \centering
%     \includegraphics[width=0.46\textwidth]{figures/quadcopter/_2-2.png}
%     \caption{Trajectory of control inputs of the quadcopter example with external wind.}
%     \label{fig:coawind}
% \end{figure}

\begin{figure}[htpb]
    \centering
    % \begin{subfigure}{0.46\textwidth}
    \includegraphics[width=0.49\textwidth]{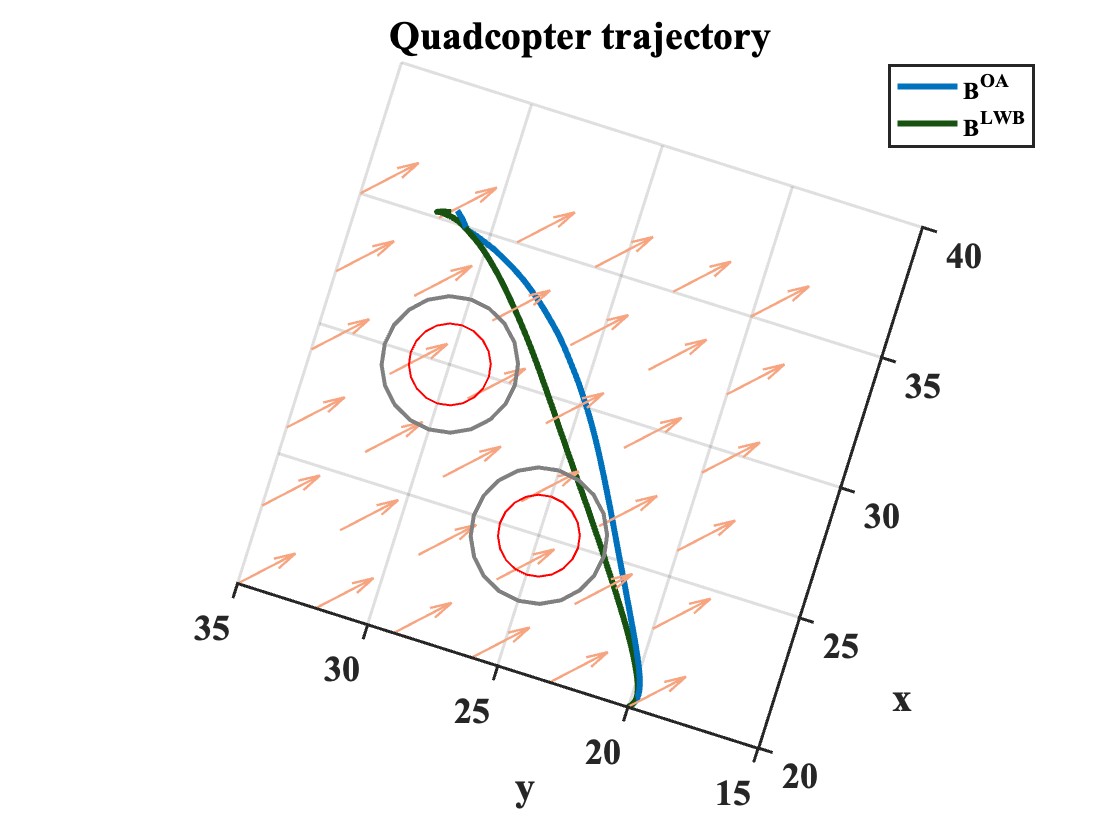}
    % \end{subfigure}
    % \hfill
    % \begin{subfigure}{0.46\textwidth}
    \includegraphics[width=0.49\textwidth]{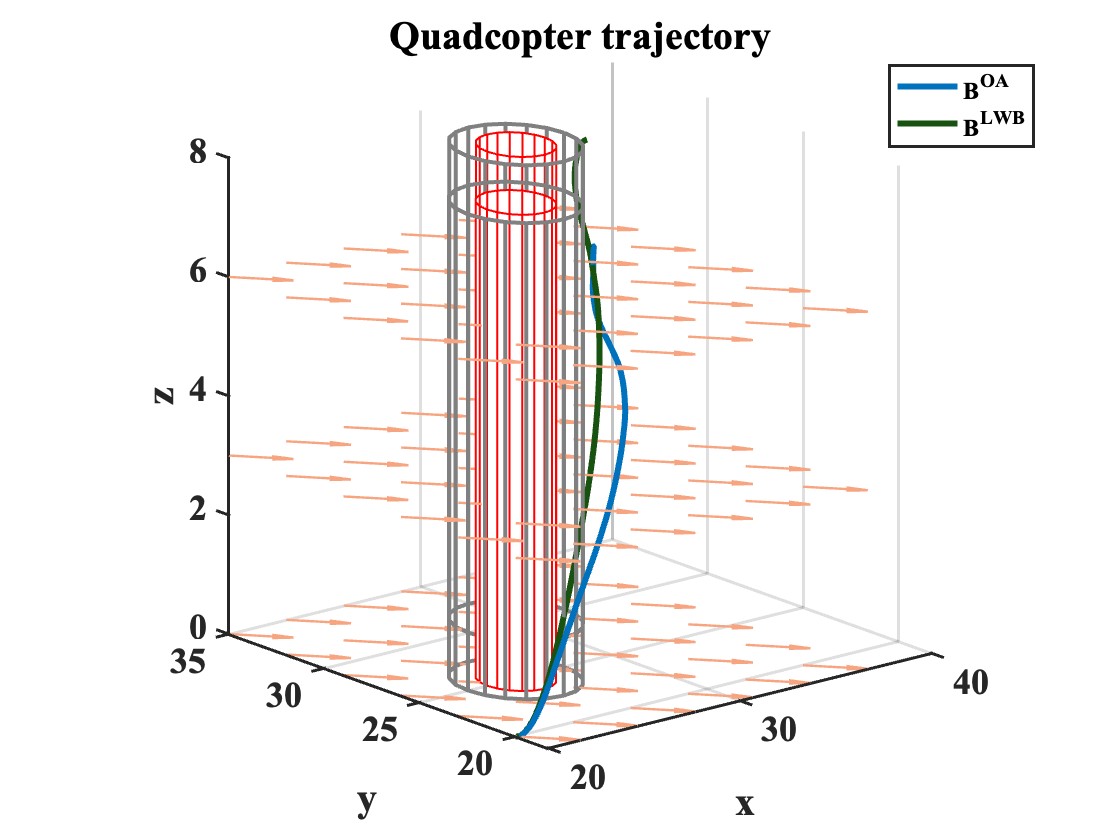}
    % \end{subfigure}
    \caption{Position trajectory of the quadcopter example with external wind. Red and blue circles correspond to the simulated and actual obstacles.}
    \label{fig:trajoawind}
\end{figure}
We compare the robustness of the controller using the bases $\mathcal{B}^{LWB}$ and  $\mathcal{B}^{OA}$ against  changes in the obstacle radii during inference and simulation. Figure~\ref{fig:trajoawind} shows the position trajectory of the quadcopter for a changed radius of obstacle during inference and simulation. Figure~\ref{fig:coawind} shows the difference in the control input before and after the change in obstacle radii for each choice of basis function. The controller with $\mathcal{B}^{LWB}$ basis cannot adapt to the changes in the obstacles, whereas the controller with $\mathcal{B}^{OA}$ basis is able to modify the control input to avoid the obstacle. This is due to the fact that the $\mathcal{B}^{OA}$ encodes the nonlinear function of distance from the obstacle which is absent in $\mathcal{B}^{LWB}$.
\begin{figure}[htpb]
\centering
    \includegraphics[width=0.49\textwidth]{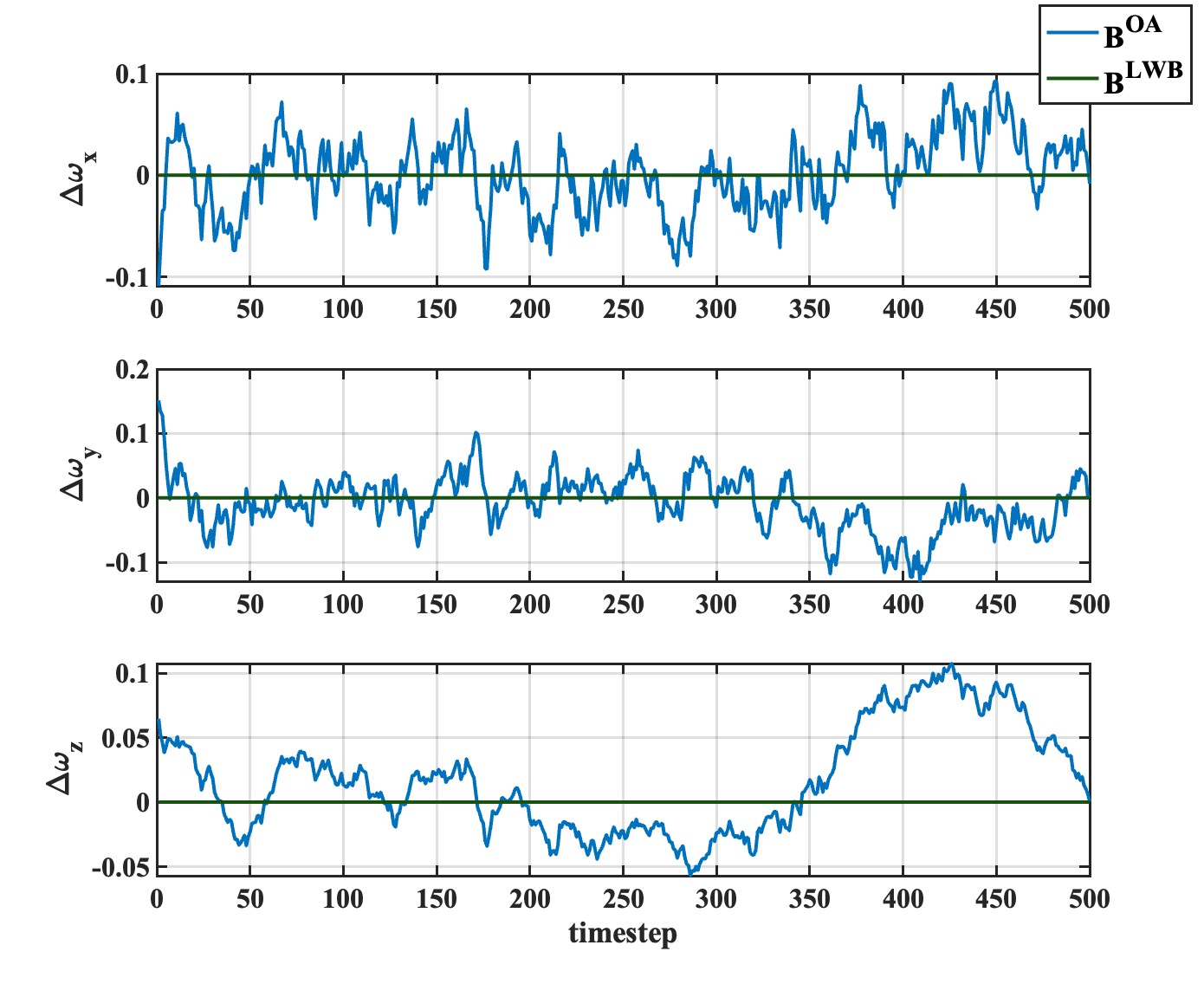}
        \includegraphics[width=0.49\textwidth]{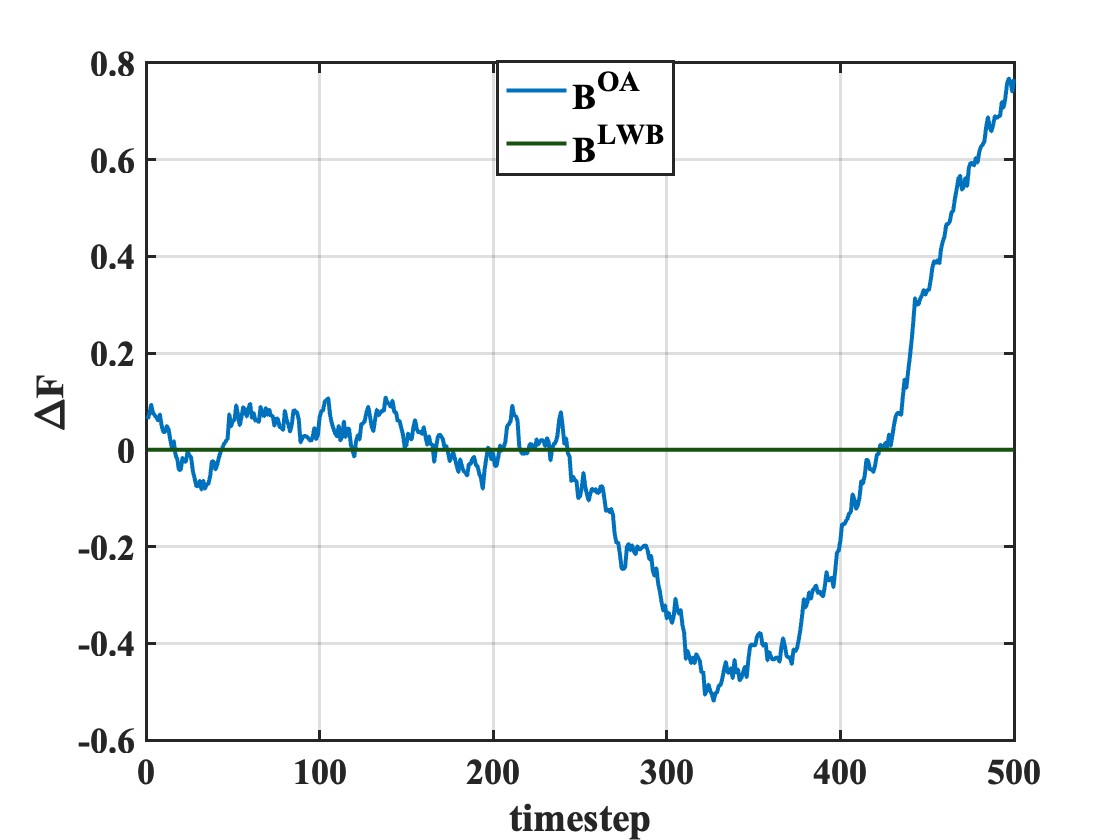}
    \caption{Comparison of the $\mathcal{B}^{LWB}$ and $\mathcal{B}^{OA}$ controllers for the effect of change in the radius of the obstacles.}
    \label{fig:coawind}
\end{figure}

\section{Conclusions and future work}
\label{sec:conclusion}

We present a parameterized inference-based approach to approximate constrained SOC. Our approach employs a barrier function to impose inequality constraints {on the states and controls,} and creates controllers satisfying given structural constraints. We establish that our approach encompasses existing algorithms as special cases, such as the LQR and the I2C algorithm. 
The numerical simulations demonstrate that our approach outperforms the ILQG for constrained control and that using factor graph optimization incurs lower average cost than unscented smoothing. Our approach can also optimize control performance while satisfying  structural constraints. Future work includes investigating structured control in a model-free setting for multi-agent systems. 
%%%%%%%%%%%%%%%%%%%%%%%%%%%%%%%%%%%%%%%%%%%%%%%%%%%%%%%%%%%%%%%%%%%%%%%%%%%%%%%%
 \section{Appendix: Equivalence of PIIC and I2C}\label{sec:eqpiici2c}
We suppress the notation $q(\tau)$ under the expectation for brevity. 
Using the block matrix inversion identity in~\cite{matcb}, the inverse term in~\eqref{eq:piic} yields
\begin{align}
   \mathbb{E}[\mathcal{B}_t(x_t)\mathcal{B}_t(x_t)^\intercal]^{-1} = \begin{bmatrix}
        \Sigma_{x_t}^{-1}&-\Sigma_{x_t}^{-1}\mu_{x_t}\\
        -\mu_{x_t}^\intercal\Sigma_{x_t}^{-1}&1+\mu_{x_t}^\intercal\Sigma_{x_t}^{-1}\mu_{x_t}
    \end{bmatrix},
    \label{eq:bxbxTinv}
\end{align}
where $\mu_{x_t}$ and $\Sigma_{x_t}$ represent the mean and covariance of $x_t$ in the smoothed state-control distribution, respectively. The second term in~\eqref{eq:piic} can be expressed as
\begin{align}
    \mathbb{E}\left [\mathcal{B}_t(x_t)u_t^\intercal\right ]
    = \begin{bmatrix}
        \Sigma_{x_tu_t}+\mu_{x_t}\mu_{u_t}^\intercal\\\mu_{u_t}^\intercal
    \end{bmatrix},
    \label{eq:xuT}
\end{align}
where $\Sigma_{x_tu_t} \in \mathbb{R}^{n_x \times n_u}$ is the cross-covariance between $x_t$ and $u_t$,  and $\mu_{u_t}$ is the mean of $u_t$ in the smoothed state-control distribution.
Substituting~\eqref{eq:bxbxTinv} and~\eqref{eq:xuT} in~\eqref{eq:piic} yields
\begin{align}
    \Theta_t^{k+1} &=\begin{bmatrix}
        \Sigma_{x_t}^{-1}\Sigma_{x_tu_t}\\
        -\mu_{x_t}^\intercal\Sigma_{x_t}^{-1} \Sigma_{x_tu_t}+\mu_{u_t}^\intercal
    \end{bmatrix}.
    \label{eq:piicupd}
\end{align}
Comparing~\eqref{eq:piicupd} and $\Theta_t = \begin{bmatrix}
    K_t&k_t
\end{bmatrix}^\intercal$ yields
\begin{align}
    K_t &= \Sigma_{x_tu_t}^\intercal\Sigma_{x_t}^{-T} = \Sigma_{x_tu_t}^\intercal\Sigma_{x_t}^{-1},\nonumber\\
    k_t &= \mu_{u_t}-\Sigma_{x_tu_t}^\intercal\Sigma_{x_t}^{-T}\mu_{x_t} 
    = \mu_{u_t}-K_t\mu_{x_t}.
    \label{eq:Kki2c}
\end{align}
The covariance $\Sigma_{\delta_t}$ can be written as
\begin{align}
  \Sigma_{\delta_t} &= \mathbb{E} [(u_t - K_t x_t - k_t)(u_t - K_t x_t - k_t)^\intercal ].
  \label{eq:cov}
\end{align}
Substituting~\eqref{eq:Kki2c} in~\eqref{eq:cov} and rearranging yields
\begin{align}
    \Sigma_{\delta_t} &=  \Sigma_{u_t} - \Sigma_{x_tu_t}^\intercal\Sigma_{x_t}^{-1}\Sigma_{x_tu_t}.
    \label{eq:Si2c}
\end{align}
Note that~\eqref{eq:Kki2c} and~\eqref{eq:Si2c} correspond to the parameter update equations for the conditional control distribution in~\cite{effsoc}. Hence, for the given assumptions on  $\mathcal{B}_t(x_t)$ and  $\Theta_t$, the PIIC and the Gaussian-I2C formulations are equivalent. Since~\cite{i2c} guarantees the equivalence of I2C to LQR {for} linear deterministic dynamics with infinitely broad priors (i.e., $\Sigma_{\eta_t} \rightarrow 0$, $\Sigma^{-1}_{\delta_t} \rightarrow 0$), we omit details of the derivation and extend the claim to the PIIC. 

% \printbibliography
\bibliographystyle{plain}
\bibliography{references_journal}
\end{document}